\documentclass[11pt,letterpaper]{article}
\usepackage[top=1in,left=1in,bottom=1in,right=1in]{geometry}
\usepackage{graphicx} 
\usepackage[colorlinks=True,allcolors=blue]{hyperref}

\usepackage{amssymb,amsmath,amsfonts}
\usepackage{amsthm}
\usepackage{xspace}

\newcommand{\set}[1]{\{#1\}}

\newtheorem{theorem}{Theorem}
\newtheorem{lemma}{Lemma}
\newtheorem{corollary}{Corollary}

\newtheorem{observation}{Observation}

\newtheorem{definition}{Definition}[section]

\title{Bounds on the Treewidth of Level-$k$ Rooted Phylogenetic Networks}
\author{Alexey Markin$^{1,\ast}$, Sriram Vijendran$^2$, Oliver Eulenstein$^3$\\
\footnotesize\textit{$^{1}$~Virus and Prion Research Unit, National Animal Disease Center, USDA-ARS}
\\
\footnotesize\textit{$^{2}$~Department of Computer Science, Iowa State University}
\\[2pt]
\footnotesize\textit{*Email: alexey.markin@usda.gov}%
}

\date{}

\begin{document}

\maketitle

\begin{abstract}
    Phylogenetic networks are directed acyclic graphs that depict the genomic evolution of related taxa. Reticulation nodes in such networks (nodes with more than one parent) represent reticulate evolutionary events, such as recombination, reassortment, hybridization, or horizontal gene transfer. Typically, the complexity of a phylogenetic network is expressed in terms of its level, i.e., the maximum number of edges that are required to be removed from each biconnected component of the phylogenetic network to turn it into a tree. Here, we study the relationship between the level of a phylogenetic network and another popular graph complexity parameter -- treewidth. We show a $\frac{k+3}{2}$ upper bound on the treewidth of level-$k$ phylogenetic networks and an improved $(1/3 + \delta) k$ upper bound for large $k$. These bounds imply that many computational problems on phylogenetic networks, such as the small parsimony problem or some variants of phylogenetic diversity maximization, are polynomial-time solvable on level-$k$ networks with constant $k$. Our first bound is applicable to any $k$, and it allows us to construct an explicit tree decomposition of width $\frac{k+3}{2}$ that can be used to analyze phylogenetic networks generated by tools like SNAQ that guarantee bounded network level.
    Finally, we show a $k/13$ lower bound on the maximum treewidth among level-$k$ phylogenetic networks for large enough $k$ based on expander graphs.
\end{abstract}

\section{Introduction}
In recent years, there has been a surge of interest in analyzing genomic data using phylogenetic networks and their variants~\cite{Kong:2022net-classes}, including hybridization networks~\cite{Mirarab:2021MSC-review,blair:2020networks-pov} and ancestral recombination graphs~\cite{Lewanski:2024ARG-era}. Popular software that infers phylogenetic networks includes tools like SNAQ~\cite{Solis:2016snaq,Solis:2017phylonetworks}, PhyloNet~\cite{Wen:2018PhyloNet}, ARGweaver~\cite{Rasmussen:2014ARGweaver}, and tsinfer~\cite{Kelleher:2019tsinfer}. However, despite many advances in the area of phylogenetic networks, both inference and analysis of networks remain computationally challenging. For example, classic phylogenetic problems like the \emph{small parsimony problem} or \emph{phylogenetic diversity maximization} are NP-hard on phylogenetic networks~\cite{Fischer:2015net-parsimony,Bordewich:2022net-pd}. This calls for the exploration of computational approaches to curb the algorithmic complexity associated with phylogenetic networks.

One such approach is parameterizing the complexity of algorithms that run on top of phylogenetic networks using \emph{treewidth}. Many classical NP-hard problems in computer science are polynomial-time solvable on graphs with bounded treewidth (see \cite{bodlaender1993tw-survey} for a survey). Recently, Scornavacca and Weller showed that multiple versions of the \emph{small parsimony problem} are polynomial-time solvable on phylogenetic networks of bounded treewidth~\cite{Scornavacca:2022smp-tw}. Further, van Iersel et al. showed a similar property for the \emph{tree containment} problem of locating a phylogenetic tree in a phylogenetic network~\cite{vanIersel:2023embedding-tw}.

In this work, we examine the relationship between the \emph{level} of a phylogenetic network, a parameter that is frequently used in phylogenetics to limit the complexity of networks, and the treewidth of the underlying unrooted graph. A rooted phylogenetic network $N$ is \emph{level-$k$} if it is sufficient to remove up to $k$ edges from each biconnected component of $N$ to turn it into a tree. We show that the treewidth of a level-$k$ network is at most $\frac{k + 3}{2}$ and we show a tighter upper-bound of $(1/3 + \delta)k$ for any small $\delta > 0$ and large $k$.  Additionally, we show that there are always level-$k$ networks with treewidth linear in $k$; in particular, we show a $\frac{k}{13}$ lower bound on the treewidth of some level-$k$ networks (for large $k$).

Our upper-bound results suggest that the networks generated by tools like SNAQ~\cite{Solis:2016snaq}, PhyNEST~\cite{Kong:2024inference}, and Squirrel~\cite{Holtgrefe:2024squirrel} with a constant network level can be efficiently decomposed into trees with small treewidth. As the level of generated networks grows~\cite{Pyron:2024higher-k}, we may still expect the treewidth to stay low and treewidth-parametrized algorithms to be efficient on higher-level networks.

\section{Preliminaries}
A \emph{(rooted phylogenetic) network} $N = (V(N), E(N))$ is a directed acyclic graph with a single source node called the root with out-degree 1, and all other nodes either have in-degree 1 and out-degree at least 2 (tree nodes), in-degree at least 2 and out-degree 1 (reticulation nodes), or in-degree 1 and out-degree 0 (leaves). 
The underlying undirected graph of a network $N$ is denoted by $G_N$. We denote the set of leaves of $N$ by $L(N)$, the set of tree nodes by $T(N)$, the set of reticulation nodes by $R(N)$, and the root by $\rho_N$.

The \emph{reticulation number} of a network $N$, denoted $r(N)$, is the total number of the incoming edges to its reticulation nodes minus the number of reticulation nodes. The reticulation number represents the number of edges that are required to be removed from $N$ to turn it into a tree.
A phylogenetic network is \emph{binary} if all tree and reticulation nodes have degree 3. Note that the reticulation number of a binary network is the number of reticulation nodes it has.

By \emph{splitting} a node $v$ of degree at least $4$, we imply replacing it with two nodes $v_1$ and $v_2$ connected by a directed edge $(v_1, v_2)$ and distributing the neighbors of $v$ among $v_1$ and $v_2$. If we split a reticulation node, we distribute the parent edges among $v_1$ and $v_2$ (at least two parents for each, including $v_1$ as a parent for $v_2$) and assign the child node to $v_2$. In contrast, if we split a tree node, we distribute the child edges among $v_1$ and $v_2$, and assign the parent node to $v_1$. Note that the split operation does not alter the reticulation number of a network.


In this work, we often refer to a directed graph $N$ while implicitly implying the undirected graph $G_N$.

\paragraph{Treewidth.} A \emph{tree decomposition (TD)} of an undirected graph $G = (V(G), E(G))$ is a tree $T$ and a mapping $B\colon V(T) \to 2^{V(G)}$ that corresponds every node $x$ of $T$ with a \emph{bag} $B(x) \subseteq V(G)$. The pair $(T, B)$ has to satisfy the following properties:
\begin{enumerate}
    \item For every $u \in V(G)$, there must exist a bag $B(x)$ that contains $u$.
    \item For every edge $\set{u, v}$ of $G$ there must exist a bag $B(x)$ that contains both $u$ and $v$.
    \item For a node $u \in V(G)$, consider the set of tree-decomposition nodes $T_u \subseteq V(T)$ that contain $u$ in their bags (i.e., $\forall x \in T_u: u \in B(x)$). Then, $T_u$ must be connected for any $u$.
\end{enumerate}

The width of a tree decomposition $(T, B)$ is the size of the largest bag $|B(x)|$ minus 1. The \emph{treewidth} of $G$, denoted $tw(G)$, is the smallest width among all tree decompositions of $G$. That is,
\[
    tw(G) := \min_{(T, B)}\bigg(\max_{x \in V(T)} |B(x)| - 1\bigg).
\]

For a directed graph $D$, we define $tw(D) := tw(G_D)$.

\paragraph{Level-$k$ networks.} 
A \emph{level} of a network $N$ is the maximum reticulation number of a single biconnected component of $G_N$. $N$ is said to be \emph{level-$k$} if its level is at most $k$.

\section{Network minors}
In this section, we restate a recent result by Coronado et al.~\cite{Coronado:2024cherry} on the relation between the level of a rooted phylogenetic network and an undirected graph minor (Lemma~\ref{lem:minor}). We provide a new (short) proof of that lemma which gives new intuition for this relation.

\begin{observation}\label{obs:ret}
    The reticulation number of a phylogenetic network $N$ is $|E(N)| - |V(N)| + 1$. 
\end{observation}
\begin{proof}
    Note that \[
    r(N) = \sum_{v \in V(N)\setminus \set{\rho_N}} (deg^{-}(v) - 1) = |E(N)| - (|V(N)| - 1),
    \]
    where $deg^{-}(v)$ denotes the in-degree of a node $v$.
\end{proof}

This observation also extends to directed acyclic graphs with more than one root (nodes with in-degree $0$). For such graphs, the reticulation number is $|E(N)| - |V(N)| + s$, where $s$ is the number of roots.

\begin{figure}
    \centering
    \includegraphics[width=0.7\textwidth]{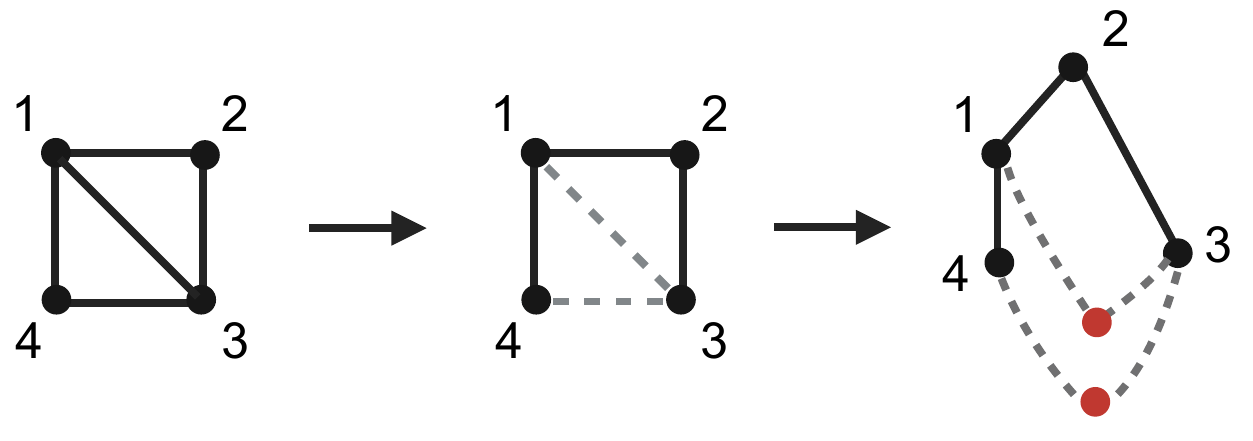}
    \caption{An example of transforming an undirected biconnected graph with $m=5$ edges and $n=4$ nodes to a phylogenetic network with reticulation number $m - n + 1 = 2$. First, we choose an arbitrary spanning tree and root it on an arbitrary node (node 2 here) -- the edges not included in the spanning tree are shown by gray dashed lines. Then, we subdivide the remaining $m - (n - 1)$ edges with reticulation nodes (shown by red circles).}
    \label{fig:transform}
\end{figure}

Next, we show the following.
\begin{lemma}[Coronado et al. \cite{Coronado:2024cherry}]\label{lem:minor}
    Let $H$ be a bi-connected graph with $m$ edges and $n$ nodes. If $H$ is a minor of a network $N$, then the level of $N$ is at least $m - n + 1$. Additionally, there always exists $N$ that has $H$ as a minor and has level exactly $k = m - n + 1$.
\end{lemma}
\begin{proof}

    Consider a network $N$ that contains $H$ as a minor. Let $M$ be a bi-connected subgraph of $N$ that can be converted into $H$ by a series of edge contractions. Let $s$ denote the number of roots of $M$ ($s \ge 1$). Then the reticulation number $r(M)$ is $|E(M)| - |V(M)| + s \ge |E(M)| - |V(M)| + 1$. Note that an edge contraction reduces the number of edges in a graph by at least $1$, and reduces the number of nodes by exactly one. That is, the $|E(M')| - |V(M')| + 1$ quantity does not increase after each edge contraction. Therefore, we have

    \[
        k \ge r(M) \ge |E(M)| - |V(M)| + 1 \ge m - n + 1.
    \]

    It remains to show that graph $H$ can be transformed into a network $N$ of level $m - n + 1$. We perform the transformation by first extracting an arbitrary spanning tree of $H$, and rooting it on an arbitrary node. Then, we subdivide any edge of $H$ not included into the spanning tree with a reticulation vertex with two "sides" of the subdivided edge pointing to it. Figure~\ref{fig:transform} shows an example of such a transformation. Since a spanning tree has $n - 1$ edges, the reticulation number resulting from such transformation is $m - n + 1$.
    
\end{proof}

Lemma~\ref{lem:minor} is a useful result in studying the phylogenetic networks. For example, using this lemma, we can easily establish when phylogenetic networks are planar or outerplanar.

\begin{corollary}[Moulton and Wu 2022~\cite{Moulton:2022planar}]\label{cor:planar}
    All networks of level 3 or less are planar, and there exists a level-4 non-planar network. Further, all level-1 networks are outerplanar.
\end{corollary}

    

Next, we use Lemma~\ref{lem:minor} to establish several new upper and lower bounds on the treewidth of level-$k$ networks.

\section{Upper bounds on treewidth of level-$k$ networks}

Janssen et al. noted that the treewidth of a levek-$k$ network is at most $k+1$~\cite{Janssen:2019DG-treewidth}. While this upper-bound is tight for level-1 networks, we can show a tighter upper-bound for $k \ge 2$:

\begin{theorem}\label{thm:ub-simple}
    The treewidth of a level-$k$ network $N$ is at most $\frac{k + 3}{2}$.
\end{theorem}
\begin{proof}
    We prove this statement by induction on $k$. For $k= 0$, the theorem holds trivially, and for $k=1$, it holds by \cite{Janssen:2019DG-treewidth}.
    
    Note that the treewidth of a graph equals the largest treewidth of its biconnected components~\cite{Bodlaender:1998k-arboretum}. Consider a biconnected component of $N$ with the maximum treewidth -- let us denote it $M$. If $G_M$ does not have any degree-3 (or higher) nodes, then $G_M$ is a cycle, and hence the treewidth of $M$ and, consequently, $N$ is $2$, and the theorem holds for any $k \ge 1$.

    Otherwise, remove any degree-3 (or higher) node from $G_M$, obtaining a graph $G'$. Note that $|E(G')| \le |E(G_M)| - 3$ and $|V(G')| = |V(G_M)| - 1$. By Lemma~\ref{lem:minor}, we can turn $G'$ into a directed network $M'$ with level $k' = |E(G')| - |V(G')| + 1 \le (|E(G_M)| - 3 - |V(G_M)| + 1) +1 \le k - 2$.
    By induction, the treewidth of $M'$ is at most $\frac{k + 1}{2}$. Further, adding a single node to a graph can increase its treewidth by at most one (e.g., this node can be added to all bags of a tree decomposition of the original graph). That is, $tw(M) \le tw(M') + 1 \le \frac{k+3}{2}$.
\end{proof}

For large $k$, we can obtain a better bound by adapting a result of Fomin and H{\o}ie~\cite{Fomin:2006pathwidth}.

\begin{lemma}[Fomin and H{\o}ie 2006]\label{lem:fomin}
    For any $\epsilon > 0$, there exists an integer $n_{\epsilon}$ such that for any graph $G$ with node degree at most $3$ and $|V(G)| \ge n_{\epsilon}$, we have $tw(G) \le (\frac{1}{6} + \epsilon)|V(G)|$ where $n_{\epsilon} = (8 / \epsilon) \cdot ln(1 / \epsilon) \cdot (1 + 1 / \epsilon^2)$.
\end{lemma}

Before we prove our main upper bound, observe that contracting degree-2 nodes does not change the treewidth.


\begin{observation}\label{obs:subdivision}
    Contracting an edge incident on a degree-2 node does not change the treewidth of a biconnected graph with more than 3 nodes.
\end{observation}
\begin{proof}
    Let $G'$ be a graph obtained by a contraction of edge $\set{u,v}$ from graph $G$. W.l.o.g. assume that $u$ is a degree-2 node, with $u$ also adjacent to a node $w$. It is well known that an edge contraction cannot increase the treewidth of a graph; hence, $tw(G') \le tw(G)$~\cite{Robertson:1986minorsII}.

    To show that $tw(G) \le tw(G')$, consider a tree decomposition $(T, B)$ of $G'$. Let $b \in B$ be a bag that contains $u$ and $w$ (it must exist since $u$ and $w$ are adjacent in $G'$). We then create a new bag, $b' = \set{u,v,w}$ and connect it to $b$ in $T$, obtaining a tree decomposition of $G$. Since $G$ is biconnected and has at least 4 nodes, $G'$ must also be biconnected on at least 3 nodes and therefore $tw(G') \ge 2$. Then, our modification of $(T, B)$ cannot increase the width of the tree decomposition.
\end{proof}

\begin{theorem}\label{thm:third}
    Let $N$ be a level-$k$ network where every maximal biconnected component has reticulation number $k \ge n_{\epsilon}$ for some fixed $\epsilon > 0$. Then, $tw(N) < (1/3 + 2\epsilon)k$.
\end{theorem}
\begin{proof}
    First, we convert $N$ into a binary network $N'$ by repeatedly splitting nodes of a degree higher than $3$ until all tree and reticulation nodes have a degree exactly $3$. Note that splitting a node does not change the network's reticulation number and cannot decrease the treewidth. To see this for the treewidth, note that if we split a node $v$ into two nodes $v_1$ and $v_2$, then one can replace $v_1$ and $v_2$ in a tree decomposition of the new graph with $v$ and obtain a tree decomposition of the original graph with the same width (or smaller).
    
    Let $M$ be a maximal biconnected component of $N'$ with the maximum treewidth. Note that we enforced the reticulation number of $M$ to be $k$. Let $n_3$ and $n_2$ denote the number of degree-3 and degree-2 nodes in $M$, respectively. By Lemma~\ref{lem:minor}, $k \ge \frac{3 \cdot n_3}{2} + n_2 - (n_3 + n_2) + 1 = \frac{n_3}{2} + 1$. From which we get $n_3 \le 2k - 2$.

    Next, we contract any edge in $M$ that is incident on a degree-2 node until either all nodes become degree-3 or the number of nodes becomes $n_\epsilon$, resulting in a graph $M'$. If all nodes in $M'$ have degree-3, then $|V(M')| \le n_3$; otherwise, $|V(M')| = n_{\epsilon}$. In both cases, $|V(M')| \le 2k - 2$.
    
    By Lemma~\ref{lem:fomin} and Observation~\ref{obs:subdivision}, we have
    \[
        tw(M) \le tw(M') \le \left(\frac{1}{6} + \epsilon\right)|V(M')| < \left(\frac{1}{3} + 2\epsilon\right)k.
    \]
\end{proof}

\section{Lower bound on the maximum treewidth of level-$k$ networks}

The lower bound can be established by examining expander graphs.

\begin{definition}[Node expansion]\label{def:vexpander}
    For a set $X \subseteq V(G)$ of a graph $G$, let $N(X)$ be the neighborhood of $X$ in $G$, i.e., the set of nodes in $V(G) \setminus X$ that are adjacent to a node in $X$. Then, for any $\alpha \in [\frac{1}{|V(G)|},1]$, the \emph{node expansion} of $G$ with parameter $\alpha$ is
    \[
        x_{\alpha}(G) := \min_{\substack{X \subseteq G\\ 0 < |X| \le \alpha |V(G)|}}{\frac{|N(X)|}{|X|}}
    \]
\end{definition}


    

\begin{lemma}[Lubotzky~\cite{Lubotzky:1994book}, Proposition 1.2.1]\label{lem:lubotzky}
    For $d \ge 5$ and some large $n$, there exists a $d$-regular graph $G$ with $x_{1/2}(G) = 1/2$.
\end{lemma}

The classical result by Grohe and Marx relates treewidth and node expansion.
\begin{lemma}[Grohe and Marx~\cite{Grohe:2009treewidth}]\label{lem:tw-exp}
    $tw(G) \ge \lfloor x_{\alpha}(G) \cdot (\alpha / 2) \cdot |V(G)| \rfloor$
\end{lemma}

Using the above results, we can show a lower bound for the maximum treewidth of level-$k$ networks.
\begin{theorem}
    For large $k$, the maximum treewidth among level-$k$ networks is at least $\frac{k}{13}$
\end{theorem}
\begin{proof}
    Let $G$ be a $5$-regular graph on $n$ nodes with $x_{1/2}(G) = 1/2$ (Lemma~\ref{lem:lubotzky} guarantees its existence for a large $n$). Then, let $N$ be a network with a single biconnected component obtained from $G$ as shown in Lemma~\ref{lem:minor}. The level of $N$ is $k = \frac{5}{2}n - n + 1$, from which we get $n = \frac{2}{3}(k-1)$. By Lemma~\ref{lem:tw-exp}, it follows
    \[
        tw(N) \ge \left\lfloor \frac{1}{2} \times \frac{1}{4} \times \frac{2}{3}(k-1) \right\rfloor \ge \frac{k-1}{12} - 1 > \frac{k}{13}.
    \]
\end{proof}

\section*{Conclusion}
Many phylogenetic networks generated in practice, particularly the hybridization networks, are level-$1$ or level-$2$ networks. These networks' treewidth is exactly 2 by Theorem~\ref{thm:ub-simple}, and therefore tree decomposition-based algorithms are directly applicable to such networks. While we can expect the level of generated networks to increase in the future, our results suggest that the treewidth will remain low as long as $k$ is low. Further investigation is required to reduce the gap between our $k/13$ and $k/3$ lower- and upper-bounds on the treewidth for large $k$, and design other practical upper-bounds on the treewidth of phylogenetic networks.

\bibliography{references}
\bibliographystyle{abbrv}

\end{document}